\documentclass[preprint,12pt]{elsarticle}

\usepackage{amssymb}
\usepackage{amsmath}
\usepackage{amsthm}
\usepackage{wrapfig}
\usepackage[linesnumbered, ruled, vlined]{algorithm2e}
\usepackage{url}
\usepackage{hyperref}

\DeclareMathOperator{\diam}{diam}
\DeclareMathOperator{\comporder}{comporder}

\journal{Discrete Applied Mathematics}

\begin{document}

\newtheorem{theorem}[equation]{Theorem}
\newtheorem{proposition}[equation]{Proposition}
\newtheorem{lemma}[equation]{Lemma}
\newtheorem{corollary}[equation]{Corollary}
\theoremstyle{definition}
\newtheorem{definition}[equation]{Definition}
\theoremstyle{remark}
\newtheorem{remark}[equation]{Remark}
\newtheorem{example}[equation]{Example}

\begin{frontmatter}



\title{A Generalization of Distance Domination} 


\author{Alicia Muth} 
\author{E. Dov Neimand} 

\affiliation{organization={Departments of Mathematics and Computer Engineering, Stevens Institute of Technology},
            addressline={1 Castle Point Terrace}, 
            city={Hoboken},
            state={NJ},
            country={United States}}

\begin{abstract}
  Expanding on the graph theoretic ideas of $k$-component order connectivity and distance-$\ell$ domination, we present a quadratic-complexity algorithm that finds a tree’s minimum failure-set cardinality i.e. the minimum cardinality any subset of the tree's vertices must have so that all clusters of vertices further away than some $\ell$ do not exceed a cardinality threshold. Applications of solutions to the expanded problems include choosing service center locations so that no large neighborhoods are excluded from service, while reducing the redundancy inherent in distance domination problems.
\end{abstract}


\begin{highlights}
\item We generalize the concept of distance domination in graphs.
\item We present an algorithm that, under the generalization, finds a minimal failure set for a tree.
\end{highlights}


\end{frontmatter}



\section{Introduction}\label{sec1}
For standard graph-theoretic notation and terminology, we refer to the book ``Graphs and Digraphs" authored by Chartrand, Lesniak, and Zhang \cite{chartrand_lesniak_zhang_2011}.
%

Given an order $n$ graph $G=(V, E)$, Haynes et al. in \cite{haynes_hedetniemi_slater_1998} define a set $S \subseteq V$ as a \textbf{distance-$\ell$ dominating set} if every vertex of $V$ is within distance $\ell$ of at least one vertex in $S$.

A \textbf{$k$-component order connectivity} set of a graph $G$ is the set $S \subseteq V$ such that all components of the graph induced by $V \setminus S$ have order less than $k$.

A \textbf{$k$-component order neighbor connectivity} set of a graph $G$ is the set $S \subseteq V$ such that all components of the graph induced by $V \setminus N_1[S]$ have order less than $k$.

The \textbf{closed-$\ell$ neighborhood} of a vertex $v \in V$, denoted $N_{\ell} [v]$, is the set $$N_{\ell} [v] = \{ u \in V \mid  d(u,v) \le \ell\}.$$

A set $F$ is a \textbf{distance-$\ell$ dominating set} if $N_{\ell}[F]=V$.
The \textbf{distance-$\ell$ domination number}, $\gamma_{\le \ell}(G)$, is the minimum cardinality of a distance-$\ell$ dominating set in $G$.

Applications of distance domination include the optimization of the placement of resources.
For example, placement of bus stops to minimize the number of blocks any individual in a given area may have to walk; or the placement of radio stations, which individually can cover a particular radius, to minimize the number of stations while still covering an entire area \cite{ameenal_bibi_lakshmi_jothilakshmi_2017}.
Additionally, there are opportunities for this parameter to take on even more properties.
For example, it is possible to use weighted vertices to represent the population of a block; or digraphs to represent one-way streets.

As distance domination problems, these placements may lead to significant overlap due to the requirement that all vertices be within $\ell$ distance of a distance dominating set. Our extension to $k$-component order neighbor connectivity allows for particular subsets of a region to be omitted from coverage, provided they are `small enough.'
These results have cost-saving benefits as they allow for minimum installations while still providing optimum service for all but some smaller clusters.
This parameter could be used to minimize the number of installations, and therefore the budget while keeping in mind the maximum distance needed to travel.

For relevant works preceding ours, in \cite{luttrell_2013}, Luttrell's dissertation fully introduces $k$-component order neighbor connectivity. Doucette et al. expand on Luttrell's work in \cite{doucette_muth_suffel_2018}.
Gross et al. \cite{gross2013survey} give us the $k$-component order neighbor connectivity definition we use here. They denote the cardinality of the smallest such $S$ with $\kappa_{nc}^{(k)}(G)$.
More information regarding the various domination-inspired parameters can be found in Luttrell et al \cite{luttrell_iswara_kazmierczak_suffel_gross_saccoman}.

\section{Initial Results}

First introduced is the definition of the vulnerability parameter, $k$, that is the focus of the paper.

Below, except where stated otherwise, we use $G = (V, E)$ to refer to a simple order-$n$ graph.

\begin{definition} 
A set $S$ is a \textbf{distance-$\ell$, $k$-component order, neighbor connective failure set} if the subgraph of $G$ induced by $V \setminus N_{\ell}[S]$ leaves all components with order less than $k$.
\end{definition}

\begin{definition} 
The \textbf{distance-$\ell$, $k$-component order connectivity number} of a graph, $\lambda_{\ell}^{(k)}(G)$, is the minimum cardinality of a distance-$\ell$, $k$-component failure set for $G$. We require the component threshold, $k$, always has $1 \le k \le n$, and the distance $\ell \geq 0$. Below, we use $k$ as a component threshold and $l$ as a distance.
\end{definition}

\begin{wrapfigure}{R}{.505\textwidth} 
\caption{\\A Minimum Distance-1, 3-Component Failure Set}
\label{fig:failedGraph}
    \centering
    \includegraphics[width=0.3\textwidth]{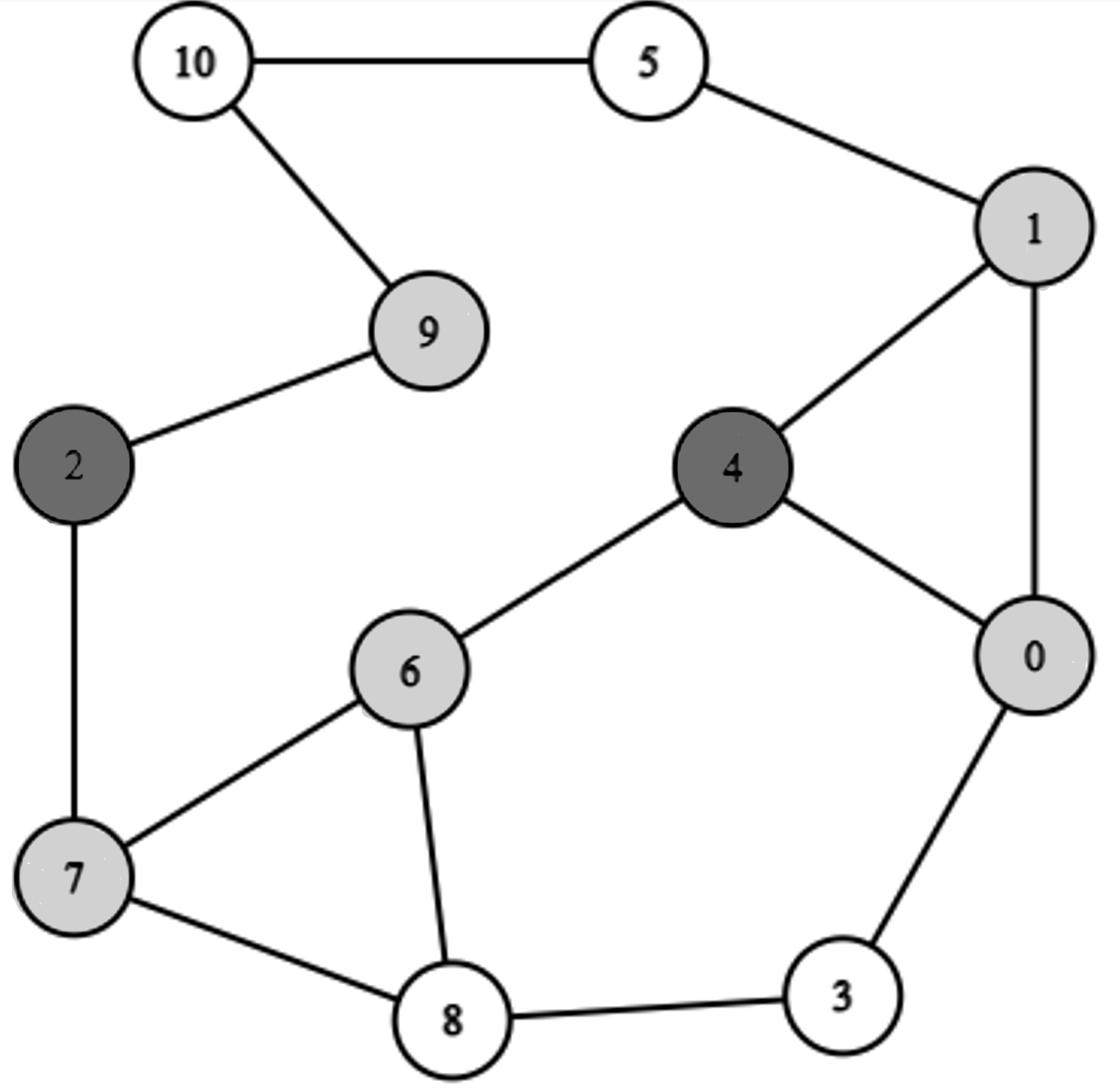}
\end{wrapfigure}

Where $k = 1$ then the problem becomes one of distance domination.  When $\ell = 0$ it becomes a problem of $k$-component order connectivity.

For distance-$\ell$, $k$-component order neighbor connectivity, a vertex $v \in V$ is said to \textbf{fail} itself and each of its $\ell$-adjacent neighbors. That is, $v$ fails all $w \in N_{\ell}[v]$.

For $F \subseteq V$, a component $H$ of $V \setminus N_{\ell}[F]$ is \textbf{failed} if $H$ has order less than $k$.

A graph $G$ is \textbf{failed} by a set $F$ if each component of the subgraph of $G$ induced by $V \setminus N_{\ell}[F]$ has order less than $k$.

In Figure \ref{fig:failedGraph} we offer an example of a minimum failure set with $k=3$ and $\ell = 1$.  The darker nodes are those in the failure set, the lighter nodes are within a distance of 1 from nodes in the failure set, and the white nodes make up components with sizes less than 3.

This parameter was first introduced in the dissertation of A. Muth using the notation $\kappa_{nc, \le \ell}^{(k)}(G)$ where the subscript $nc$ stands for neighbor connectivity\cite{muth}. 

For simplicity, we refer to a distance-$\ell$, $k$-component order neighbor connectivity failure set as a \textbf{failure set} and if there is no failure set with smaller cardinality, a \textbf{minimum failure set}. 

We consider the relationship between the cardinality of a minimum failure set, $F$, and the values of $k$ and $\ell$. As the order of surviving components increases, the cardinality of the minimum failure set decreases.

\begin{proposition} 
\label{prop:com thresh inequality}
Let $j$ and $k$ be component thresholds with $j \le k$. Then
$$\lambda_{\ell}^{(k)}(G) \le \lambda_{\ell}^{(j)}(G).$$
\end{proposition}

Similarly, as the distance increases, minimum failure-set cardinality may decrease.

\begin{proposition} 
\label{prop:ells}
Let $m$ and $\ell$ be distances with  $m \le \ell$.
Then, $$\lambda_{\ell}^{(k)}(G) \le \lambda_m^{(k)}(G).$$
\end{proposition}

Notice that the next two propositions follow directly from the definition of the various parameters and the preceding results.

\begin{proposition} 
\label{prop:domination}
The distance domination number, $\gamma_{\le \ell}(G)$, of a graph is $\lambda_\ell^{(1)}(G)$.  Propositions \ref{prop:com thresh inequality} and \ref{prop:ells} combine for the following:

$$\lambda_{\ell}^{(k)}(G) \le \lambda_{\ell}^{(1)}(G) = \gamma_{\le \ell}(G)$$.
\end{proposition}

\begin{proposition}
\label{prop:kcomp}
For distance $\ell = 1$, a set $F$ is a distance-$1$, $k$-component order neighbor connectivity failure set, if and only if it is a $k$-component order neighbor connectivity failure set, i.e.,
$$ \lambda_{\ell}^{(k)}(G) \le \lambda_{1}^{(k)}(G) = \kappa_{nc}^{(k)}(G).$$
\end{proposition}

If the maximum distance between any two vertices of the graph, the diameter, is less than $\ell$ then any nonempty failure set will fail the graph.

\begin{proposition}
If $0 < \diam(G) \le \ell$, then $\lambda_{\ell}^{(k)}(G) = 1$.
\end{proposition}

\begin{proposition}
Let $G$ be disconnected and $G = \bigcup\limits_{i=1}^{N} G_i$.
Then the order of a minimum failure set of the graph $G$ is the sum of the orders of the minimum failure sets for each $G_i$, i.e.,
$$\lambda_{\ell}^{(k)}(W) = \sum_{i=1}^N \lambda_{\ell}^{(k)}(W_i).$$
\end{proposition}

\section{A Minimum Failure Set Algorithm}\label{sec2}

Finding a minimum failure set for distance domination on an arbitrary simple graph is an NP-hard problem, \cite{haynes1998domination}.  
The distance domination problem is a distance-$\ell$ $1$-component order, neighbor connectivity failure set problem, and is therefor NP-hard. It follows that the more generic $k$-component problem is NP-hard as well.  Consequently, we limit the scope of our algorithmic efforts to trees.

We use $T = (V, E)$ to indicate a rooted tree arbitrarily rooted at $r$, and for some $v \in V$ we use $T_v = (V_v, E_v)$ to denote the subtree of $T$ rooted at $v$. A node
 $v$'s neighbor is the parent if it is on the unique path from $v$ to $r$, and a child, $c \in C_v$, if it is not.  

In order to construct an algorithm that finds a minimum failure set, we need to compute the order of a tree component.  We do this as follows.

\begin{definition}
Let $v \in V$ and $F \subseteq V$. If there exists a $w \in F$ such that $v \in N_\ell[w]$ we say that component order $\comporder(F, v) = 0$.  If $v$ is not in $N_\ell[F]$ with children $C_v$, we recursively define $\comporder(F, v) = 1 + \sum_{c \in C_v} \comporder(F, c)$. Where $F$ is unambiguous, we omit it.
\end{definition}

\begin{lemma}
A component $K$ of the graph induced by $V \setminus F$ has an order of $\comporder(v)$ where $v$ is the root of the smallest subtree containing $K$.
\end{lemma}

\begin{proof}
The component $K$ is a tree itself and the component order of the root $r_K$ of $K$ is the order of $K$.  Proof by induction on the order of $K$ is immediate.
\end{proof}

\begin{remark}
Consider Algorithm \ref{algo:setSelectons}.  We prove in Theorem \ref{theorem:grand finaly} that the algorithm constructs a minimum failure set $F$. The algorithm takes in a vertex $v \in V$ as input.  Any vertex is acceptable, and every vertex of $V$ gets called by the recursive action of the algorithm.  The user of the algorithm only calls the algorithm on $r$, initializing $F \gets \emptyset$ when doing so.  For implementation details beyond those presented here, see \cite{Neimand_Code_Written_An_2022}.
\end{remark}

\begin{algorithm}
\DontPrintSemicolon

\KwIn{
\begin{itemize}
    \item Constants $k, \ell \in \mathbb N$ as described above. 
    \item A vertex $v$ of a tree $T$ with root $r$.  
    \item A global $F \subseteq V$.  If $v=r$, then $F \gets \emptyset$.
\end{itemize}
}
\KwOut{

\begin{itemize}
    \item $F \subseteq V$ is the failure set after the algorithm has run on $r$.
\end{itemize}

}
    
\lForEach{$c$ \upshape child of $v$ \label{algo.line:recursive}}{
    Run Algorithm \ref{algo:setSelectons} on $c$
}

Set component orders for $T_v$. \label{algo.line:set comp order}

\lIf{$v = r$  {\bf and} $\exists u \in N_\ell[r]$ \upshape such that $\comporder(u) \ge k$ \label{algo.line:root plan}}{
    Add $v$ to $F$.
}

\ElseIf{$\exists u$ \upshape such that $u$ is an $\ell$-generation descendent of $v$ with $\comporder(u) \ge k$ \label{algo.line:fail filter}}{
    Add $v$ to $F$. \label{algo.linen:fail f and mark}
}

\caption{Selects Vertices for Failure}\label{algo:setSelectons}
\end{algorithm}

\begin{example}
\label{exmp: run through}
We will run the algorithm on a small example problem with $\ell := 1$ and $k := 1$.  The root of our graph will be $r=a$.  We'll set $a$'s children to $a^2, ab$ and $ac$.  We'll set $a^2$'s children to $a^3$ and $a^2 b$.  We'll set $ab$'s child to $aba$ and $aba$'s child to $aba^2$.  See Figure  \ref{fig:example graph}.
\end{example}

\begin{wrapfigure}{R}{0.3\textwidth} 
\caption{Example \ref{exmp: run through}}
\label{fig:example graph}
    \centering
    \includegraphics[width=0.3\textwidth]{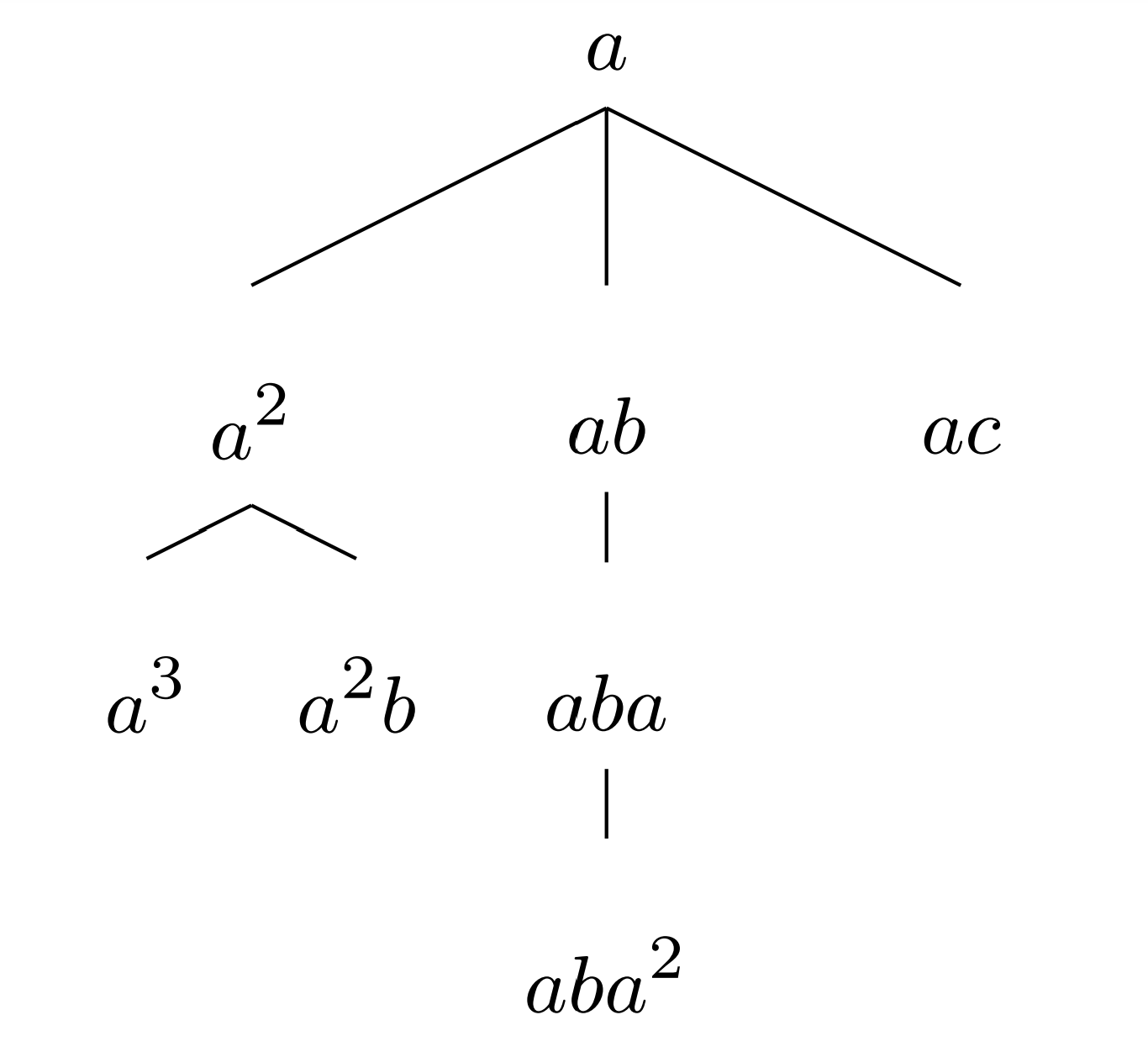}
\end{wrapfigure}

 Always before calling the algorithm on the root, we set $F \gets \emptyset$.  That is to say, the failure set is empty, no vertices have been marked for failure. We call Algorithm \ref{algo:setSelectons} on $v \gets r=a$. On Line \ref{algo.line:recursive} we immediately call Algorithm \ref{algo:setSelectons} on $a^2$, $ab$, and $ac$.  We'll look at each of these in turn and, without loss of generality, begin with $v \gets a^2$.  On Line \ref{algo.line:recursive} we call Algorithm \ref{algo:setSelectons} on $a^2$'s children and proceed to $v \gets a^3$.  

Since $a^3$ has no children, we proceed to Line \ref{algo.line:set comp order} with $\comporder(a^3) \gets 1$. On Line \ref{algo.line:root plan} we note that $a^3 \ne r$ and proceed to Line \ref{algo.line:fail filter}.  Since $a^3$ has no 1-generation descendants (children), the condition is false, and we complete Algorithm \ref{algo:setSelectons} on $a^3$. 

We resume Algorithm \ref{algo:setSelectons}'s progress for $v = a^2$ and proceed to call Algorithm \ref{algo:setSelectons} on the next child of $a^2$, namely $a^2 b$.  Algorithm \ref{algo:setSelectons} runs on $a^2 b$ exactly as it did on $a^3$ so we omit the details.

Once Line \ref{algo.line:recursive} is completed for $v = a^2$ we continue to Line \ref{algo.line:set comp order} where $\comporder(a^3) \gets 1, \comporder(a^2b) \gets 1$ and $\comporder(a^2) \gets 1 + 1 + 1 = 3$. On Line \ref{algo.line:root plan} we note that $a^2 \ne r$ so we proceed to Line \ref{algo.line:fail filter} and check if there exists a 1-generation descendent of $a^2$ with a component order greater than or equal to 1, and find, without loss of generality, that $a^3$ has $\comporder(a^3) \ge 1$.  Therefore we add $a^2$ to $F$, equivalently $ F \gets F \cup \{a^2\}$. The vertices $a^3, a^2b, a^2$, and $a$ are all in $N_\ell[a^2]$ and near a failed vertex.  We have completed Algorithm \ref{algo:setSelectons}'s run on $a^2$ and return to $v = a$.  

Algorithm \ref{algo:setSelectons} was on Line \ref{algo.line:recursive} for $v=a$, and now proceeds to call Algorithm \ref{algo:setSelectons} on $ab$, which in turn immediately calls Algorithm \ref{algo:setSelectons} on $aba$ calling it on $aba^2$.  Since $v = aba^2$ is a leaf, the algorithm runs just as it did on $a^3$.  On completion of $v= aba^2$ we return to $v = aba$.  With $aba$ having no other children, the algorithm proceeds to Line \ref{algo.line:set comp order} setting $\comporder(aba^2) \gets 1$ and then $\comporder(aba) \gets 1+1 = 2$. Checking the condition on Line \ref{algo.line:root plan} we note that $aba \ne r$, and on Line \ref{algo.line:fail filter} we find that $aba$ does have a 1-generation descendent, $aba^2$ with $\comporder(aba^2) \ge k = 1$, so we add $aba$ to $F$ giving $F = \{a^2, aba\}$.  The algorithm completes its run on $aba$ and we return to the run $v = ab$.

The algorithm has completed its recursive calls to all of $ab$'s children, so we proceed to Line \ref{algo.line:set comp order}.  Since $aba^2$, $aba$, and $ab$ are all in $N_1[aba]$ they have component orders of 0. Proceeding to lines \ref{algo.line:root plan} and \ref{algo.line:fail filter}, the conditions are not met and the algorithm completes its run on $ab$, returning to $a$.

On Line \ref{algo.line:recursive}, $a$ has one last child, $ac$.  We call Algorithm \ref{algo:setSelectons} on $ac$, which is a leaf so the algorithm runs as it did on $a^3$, and we continue to Line \ref{algo.line:set comp order} for $v = a$.  Note, every vertex in the subtrees rooted at $a^2$ and  $ab$ is in the neighborhood of some $f \in F$ so all those vertices have component order 0, as does $a$ itself.  All that remains is $\comporder(ac) \gets 1$.  On Line \ref{algo.line:root plan} we note $a = r$ and $ac \in N_1[a]$ with $\comporder(ac) \ge 1$, so we fail $a$ and skip checking the \textbf{else if} condition.  The algorithm is complete leaving $F = \{a, a^2, aba\}$ and $V \subseteq N_1[F]$.

If the Algorithm has only run on a subtree rooted at $v$, and that subtree contains a component with order at least $k$, that component must be failed by some vertex outside the subtree.  The following Lemma and Theorem guarantee this.

\begin{lemma}
\label{lemma:proof of dov's algo with induction}
Let $v \in V$. If Algorithm \ref{algo:setSelectons} has run on $v$, and there exists a $u \in V_v$ with $\comporder(F, u) \ge k$, then $u \in N_{\ell - 1}[v]$.
\end{lemma}

\begin{proof}
We will prove by induction.  For the base case, we will apply Algorithm \ref{algo:setSelectons} to a leaf, $v \in V$.  Let $u \in V_v = \{v\}$.  We have $u = v$ with $\comporder(u) = 1$. If $k = 1$, then for any $\ell \ge 1$, we have $u \in \{v\} \subseteq N_{\ell-1}[v]$, the desired result. If $\ell = 0$ then the \textbf{else if} condition on line \ref{algo.line:fail filter} was met and $v$ was added to the failure set, insuring $\comporder(v) = 0$.  If $k>1$, then the desired result is achieved trivially.

For the inductive assumption, Line \ref{algo.line:recursive} calls Algorithm \ref{algo:setSelectons} on all the children of $v$. We assume that for any child $c$ of $v$, if $u \in V_c$ and $\comporder(u) \ge k$ then $d(u,c) < \ell$. Equivalently, $d(u, v) < \ell + 1$.  If $d(u,v) < \ell$, or if there is no such $u$, then the desired result is achieved, so we will assume that before the algorithm proceeds to Line \ref{algo.line:set comp order}, their exists a $u$ such that $d(u,v) = \ell$ and $\comporder(u) \ge k$. Line \ref{algo.line:fail filter}, or if $v=r$ Line \ref{algo.line:root plan}, fails $v$ resetting $\comporder(u) \gets 0$. Once the algorithm has run, no such $u$ exists.
\end{proof}

\begin{proposition}
\label{prop:is failure set}
Let Algorithm \ref{algo:setSelectons} have run on $r$. The set $F \subseteq V$ is a failure set for $T$.
\end{proposition}

\begin{proof}
By Lemma \ref{lemma:proof of dov's algo with induction} we have that any $v \in V$ with $\comporder(v) \ge k$ is in $N_{\ell-1}[r]$. Algorithm \ref{algo:setSelectons} Line \ref{algo.line:root plan} guarantees that any such $v$ was assigned $\comporder(v) \gets 0$. 
\end{proof}

Next, we would like to show that for any arbitrary failure set of the tree, a surjective mapping may be constructed from that failure set, to the set constructed by the Algorithm.

Algorithm \ref{algo:Check Failure} plays a roll in our proof that the failure set $F$, created by Algorithm \ref{algo:setSelectons}, is a minimum failure set.  It is not meant to be used for any other purpose. When called on $r$, Algorithm \ref{algo:Check Failure} builds a mapping $M \mid W \to V$. Before Algorithm \ref{algo:Check Failure} is called on $r$, this mapping is initialized to the identity function. The mapping is a global variable whose value changes without being explicitly returned or passed.  With each recursive call to Algorithm \ref{algo:Check Failure}, $M$ may be modified resulting with its final value a surjective mapping onto $F$ (Lemma \ref{lemma:surjective}).

\begin{definition}
For a mapping $M \mid W \to V$, the image of $M$ is denoted $M(W)$. We also define $M_v(W) := (M(W) \cap V_v) \setminus \{v\}$ after Algorithm \ref{algo:Check Failure} has run on all of $v$'s children.
\end{definition}

\begin{algorithm}
\DontPrintSemicolon

\KwIn{
\begin{itemize}
    \item $v \in V$
    \item $W \subseteq V$
    \item Global $M \mid W \to V$. If $v = r$, then $M \gets$ the identity function.
\end{itemize}
}

\lForEach{$c$ \upshape child of $v$ \label{algo.line:check recursive}}{
    Run Algorithm \ref{algo:Check Failure} on $c$.
}

Set component orders for $T_v$ with the failure set $M_v(W)$. \label{algo.line:check set comp size}

\If{$\forall u :$ \upshape $u$ is an $\ell$-gen. descendent of $v$, $\comporder(M_v(W), u) < k$, $v \ne r$ \label{algo.line:check fail filter}}{
    $\forall w \in W:M(w) = v$, reset $M(w) \gets v$'s parent. \label{algo.line:check move up}
}

\caption{Constructs a Mapping to F}\label{algo:Check Failure}
\end{algorithm}

\begin{definition}
For convenience, if the statement on line \ref{algo.line:check move up} is reached and for all $w$, $M(w)$ is reset to $v$'s parent $p$, we say that \textbf{$v$'s fail is moved up}.
\end{definition}

\begin{lemma}
\label{lemma:moving g doesn't change fail}
Let Algorithm \ref{algo:Check Failure} run on $v \in V$. If $W$ is a failure set, then $M(W)$ is a failure set.
\end{lemma}

\begin{proof}
Let's assume $W$ is a failure set. The starting value of $M$ is the identity function, so we need only verify that when $M$ is changed, on Line \ref{algo.line:check move up}, no new components of order greater than or equal to $k$ are created. Note, $\comporder(M_v(W), u) \le \comporder(M(W), u)$, so by considering $\comporder(M_v(W), u)$ instead of $\comporder(M(W), u)$ there is no risk of creating a new component of order greater than or equal to $k$.

The $M(W)$ components of $T$ that may be affected by moving some $v$'s fail can be divided into two categories. Those to whom the only path to $v$ is through $p$, the parent of $v$, and those through whom the only path is through a child of $v$.

Vertices connected to $v$ through $p$ can have their $M(W)$-component order only reduced when $v$'s fail is moved up, either to 0 because those vertices are in $N_\ell[p]$, or to less than what they were because a descendent was set to component order 0. 

A vertex $u$ connected to $v$ through $c$, a child of $v$, can fall into 3 disjoint sub categories: $d(u,v) < \ell$, $d(u,v) = \ell$, and $d(u,v)> \ell$.  Any $u \in N_{\ell - 1}[v]$ still has $\comporder(M(W), u) = 0$ after the change since $u \in N_\ell[p]$. Any $u$ with $d(u,v) > \ell$ remains unchanged in component order since $u$'s descendants are unchanged in their component orders.  

We will now consider a vertex $u$ with $d(u,v) = \ell$, that is, an $\ell$-generation descendent of $v$. The algorithm looks at the component order of $u$ as though $v$ and its ancestors, are not in $M(W)$, which tells us what the component order will be if $v$'s fail were moved up, as $u \in N_\ell[v] \setminus N_\ell[p]$.  If the component's order will be greater than $k$, then the condition on Line \ref{algo.line:check fail filter} is false and the statement on Line \ref{algo.line:check move up} is never reached. Consequently, $v$'s fail is not moved up.  If we'll have $\comporder(u) < k$, then $v$'s fail is moved up.  In either case, all components continue to have order less than $k$.
\end{proof}





\begin{lemma}
\label{lemma:v in M(W)}
Let $W$ be a failure set, and Algorithm \ref{algo:Check Failure} have run on $v \in V$. If for some $u$ that is an $\ell$-generation descendent of $v$, we have $\comporder(M_v(W), u) \ge k$, then $v \in M(W)$.
\end{lemma}

\begin{proof}
Lemma \ref{lemma:moving g doesn't change fail} tells us $M(W)$ is a failure set, so there exists $g \in M(W)$ such that $\comporder(M_v(W) \cup \{g\}, u) < \comporder(M_v(W), u)$. It follows that $g \in (V \setminus V_v) \cup \{v\}$ and there exists a $w$ that is a descendent of $v$, with $w \in N_{\ell}[g]$. For all such $w$, all paths from $g$ to $w$ go through both $v$ and $u$. Therefore, $d(w, g) \ge d(u,g) \ge d(u,v) = \ell$. Equalities exist only when $w=u$ and $g = v$, the desired result.
\end{proof}

\begin{definition}
Let $F_v$ be the failure set generated by Algorithm \ref{algo:setSelectons} when run on all the descendants of $v$.
\end{definition}

\begin{lemma}
\label{lemma:surjective}
Let $W \subseteq V$ be a failure set and Algorithm \ref{algo:Check Failure} have run on $r$ and $W$ generating $M$. Let Algorithm \ref{algo:setSelectons} have run on $r$ generating $F$.
\begin{enumerate}
    \item \label{lemma.item:MW in Fr} $M(W) \subseteq F \cup  \{r\}$, and
    \item \label{lemma.item:surjective} $M(W) \supseteq F$.
\end{enumerate} 
\end{lemma}

\begin{proof}
We will prove with induction. For the base case, let $v \in V$ be leaf. 

(\ref{lemma.item:MW in Fr} $\subseteq$) Let $v \in M(W)$. If we falsely assume $v \ne r$ then the condition in Algorithm \ref{algo:Check Failure} Line \ref{algo.line:check fail filter} is trivially true and Line \ref{algo.line:check move up} insures that $v$ is not in $M(W)$, a contradiction.  We assume $v = r$. Then $v \in F \cup \{r\}$.

(\ref{lemma.item:surjective} $\supseteq$) Let $v \in F$.  If $\ell = 0$, then $k=1$, and since $W$ is a failure set with $N_0(W)=W$, we have $v \in W$. Note that Algorithm \ref{algo:Check Failure} does not move $v$ up, giving $v \in M(W)$. If $\ell \ge 1$, there are no $\ell$-generation decedents of $v$. We arrive at $v = r$ and $k = 1$.  Since $W$ is a failure set, $W = \{v\}$. Under these circumstances, Algorithm \ref{algo:Check Failure} did not move $v$'s fail up and $M=id \Rightarrow M(W) = W = F$.

For the inductive hypothesis, we assume that for any non-leaf vertex $v \in V$, if algorithms \ref{algo:setSelectons} and \ref{algo:Check Failure} have run on a descendent $u$, then if $u \in M(W) \iff u \in F$, or equivalently, $F_v = M_v(W)$.

(\ref{lemma.item:MW in Fr} $\subseteq$) Let $v \in M(W)$ after Algorithm \ref{algo:Check Failure} has run on $v$.  This means $v$'s fail was not moved up and there is a $w\in W$ such that $M(w) = v$. The condition in Algorithm \ref{algo:Check Failure} Line \ref{algo.line:check fail filter} was false. If $v \ne r$, there exists $u$ that is an $\ell$-generation descendent of $v$ with $\comporder(M_v(W), u) \ge k$. Since $\comporder(F_v, u) = \comporder(M_v(W), u)$, we may conclude that when Algorithm \ref{algo:setSelectons} Line \ref{algo.line:fail filter} arrived at $v$, the condition was true and $v$ was added to $F$ achieving the desired result. If $v = r$ then the desired result is immediate.

(\ref{lemma.item:surjective} $\supseteq$) Let $v \in F$. If $v \ne r$, then Algorithm \ref{algo:setSelectons} found a $u$ that is an $\ell$-generation descendent of $v$ with $\comporder(F_v, u) \ge k$. By the inductive assumption, $\comporder(M_v(W), u) \ge k$.  By Lemma \ref{lemma:v in M(W)}, $v\in M(W)$. If $v = r$ then Algorithm \ref{algo:setSelectons} failed $r$ because there exists a $u \in N_\ell[r]$ such that $\comporder(F_r, u) \ge k$.  Using the inductive assumption together with Lemma \ref{lemma:moving g doesn't change fail}, we conclude that there exists a $g \in M(W) \setminus M_r(W)$ that lowers the component order of $u$.  The only element that could be is $r$. We conclude $v = r \in M(W)$.
\end{proof}

We arrive at our final results.

\begin{theorem}
\label{theorem:grand finaly}
The set $F$ of all vertices marked as failed after Algorithm \ref{algo:setSelectons} has been called on $r$ is a minimum failure set for $T$.
\end{theorem}

\begin{proof}
Lemma \ref{prop:is failure set} gives $F$ is a failure set. Lemma \ref{lemma:surjective}.\ref{lemma.item:surjective} tells us that for any failure set $W$, there exists a surjective mapping from $W$ to $F$, $\vert F \rvert \le \vert W \rvert$.

\end{proof}

\begin{theorem}
The computational complexity of Algorithm \ref{algo:setSelectons} run on $r$ is $O(n^2)$.  
\end{theorem}
\begin{proof}
All $n$ descendants of $r$ are called by Line \ref{algo.line:recursive}. The complexity is $O(n \cdot k)$ where $k$ is the complexity of each vertex's operations. It remains to compute $k$.

The complexity of Line \ref{algo.line:set comp order} is at most $O(n)$ since each vertex has at most $n$ descendants over which the recursive component order computation iterates. 

Lines \ref{algo.line:root plan}, \ref{algo.line:fail filter}, and \ref{algo.linen:fail f and mark}, along with the \textbf{then} statement in \ref{algo.line:root plan} are similarly $O(n)$.  This gives the complexity of Algorithm \ref{algo:setSelectons} as $O(n\cdot (n+ n + n + n + n)) = O(n^2)$.  
\end{proof}

\subsubsection*{Acknowledgments}

Thanks to the Stevens/Seton Hall Graph Theory Group for their continued support.

\bibliographystyle{plainurl} 

\bibliography{bibliography}

\end{document}